\documentclass[a4paper]{article}

\usepackage{graphicx}
\usepackage[shortlabels]{enumitem} 
\usepackage{a4wide}
\usepackage{amsmath,amsthm,amsfonts}
\usepackage{amssymb}
\usepackage{mathtools}
\usepackage{graphicx}
\usepackage{booktabs}
\usepackage{xcolor}
\usepackage{hyperref}
\usepackage{graphicx}
\usepackage[numbers]{natbib}
\usepackage[capitalize]{cleveref}
\usepackage{tikz}
\usepackage{xspace}
\usepackage{mathdots}
\usepackage[singlelinecheck=false]{caption}
\usepackage{subcaption}
\usepackage{cancel}
\usepackage{todonotes}
\usepackage{authblk}
\usepackage{algorithm, algpseudocodex}

\usetikzlibrary{tikzmark,decorations,arrows,arrows.meta,petri,topaths,backgrounds,shapes,positioning,fit,calc,decorations.pathreplacing,patterns,intersections}
\usetikzlibrary{positioning}
\usetikzlibrary{shapes.geometric}
\usetikzlibrary{graphs}
\usetikzlibrary {graphs.standard}
\tikzstyle{nodeA} = [draw=blue!65!black, circle, fill=blue!65!black, minimum size=1ex, inner sep=1pt, text centered, align=center]
\tikzstyle{nodeB} = [draw=green!60!black, circle, fill=green!60!black, minimum size=1ex, inner sep=1pt, text centered, align=center]

\newtheorem{theorem}{Theorem}
\newtheorem{lemma}[theorem]{Lemma}

\newtheorem{obs}[theorem]{Observation}

\newtheorem{definition}{Definition}

\crefname{table}{Table}{Tables}
\crefname{figure}{Figure}{Figures}
\crefname{definition}{Definition}{Definitions}
\crefname{theorem}{Theorem}{Theorems}
\crefname{lemma}{Lemma}{Lemmas}
\crefname{claim}{Claim}{Claims}
\crefname{obs}{Observation}{Observations}
\crefname{prop}{Proposition}{Propositions}
\crefname{corollary}{Corollary}{Corollaries}
\crefname{example}{Example}{Examples}
\crefname{section}{Section}{Sections}
\crefname{subsection}{Subsection}{Subsections}
\crefname{algorithm}{Algorithm}{Algorithm}

\newcommand{\NP}{\textsf{NP}\xspace}
\newcommand{\coNP}{\textsf{coNP}\xspace}
\newcommand{\poly}{\textsf{P}}

\newcommand{\local}{\mathit{local}}

\begin{document}

\title{Complexity of Perfect and Ideal Resilience Verification in Fast Re-Route Networks}

\author[1,2]{Matthias Bentert}
\author[1,3]{Esra Ceylan-Kettler}
\author[3]{Valentin H\"{u}bner}
\author[1,4]{\\Stefan Schmid}
\author[5]{Ji\v{r}\'{\i} Srba}

\affil[1]{TU Berlin, Germany}
\affil[2]{University of Bergen, Norway}
\affil[3]{Institute of Science and Technology Austria, Austria}
\affil[4]{Fraunhofer SIT, Germany}
\affil[5]{Aalborg University, Denmark}

\date{}

\maketitle 

\begin{abstract}
To achieve fast recovery from link failures, most modern communication networks feature fully decentralized fast re-routing mechanisms. These re-routing mechanisms rely on pre-installed static re-routing rules at the nodes (the routers), which depend only on local failure information, namely on the failed links \emph{incident} to the node.
Ideally, a network is \emph{perfectly resilient}: the re-routing rules ensure that packets are always successfully routed to their destinations as long as the source and the destination are still physically connected in the underlying network after the failures.
Unfortunately, there are examples where achieving perfect resilience is not possible. Surprisingly, only very little is known about the algorithmic aspect of when and how perfect resilience can be achieved. 

We investigate the computational complexity of analyzing such local fast re-routing mechanisms. Our main result is a negative one: we show that even checking whether a \emph{given} set of static re-routing rules ensures perfect resilience is \coNP-complete. 
We also show the \coNP-completeness of the so-called \emph{ideal resilience}, a weaker notion of resilience often considered in the literature.  
Additionally, we investigate other fundamental variations of the problem. In particular, we show that our \coNP-completeness proof also applies to scenarios where the re-routing rules have specific patterns (known as \emph{skipping} in the literature). 

On the positive side, for scenarios where nodes do not have information about the link from which a packet arrived (the so-called in-port), we present linear-time algorithms for both the verification and synthesis problem for perfect resilience.
\end{abstract}

\section{Introduction}

Communication networks form a critical backbone of many distributed systems, whether they are running in enterprises, in data centers, or are geographically distributed across the Internet.
In order to meet their stringent availability requirements, modern networks feature local fast re-routing mechanisms:
each router (henceforth called node) has conditional packet forwarding rules which depend on the status of its links. These rules are configured \emph{ahead of time}, without knowledge of possible link failures, and allow routers to forward packets along alternative links in case of failures, in a fully decentralized manner.

However,  configuring such conditional forwarding rules to provide a high resilience, even under multiple link failures, is algorithmically challenging, as the rules can only depend on \emph{local information}: a node is not aware of possible additional link failures downstream, in other parts of the network. Hence, if no care is taken, the local re-routing rules of different nodes can easily result in forwarding loops.

Ideally, a network and its local fast re-routing mechanism provide \emph{perfect resilience}: the local re-routing rules ensure that as long as the source is still connected to the destination \emph{in the underlying network} after the failures, then the packet is also successfully routed to the destination \emph{on the network layer}. Already at ACM PODC 2012, Feigenbaum et al.~\cite{feigenbaum2012brief} gave an example which shows that certain networks inherently cannot be configured to provide perfect resilience. 

On the positive side, it is known that under a single link failure, it is always possible to successfully route packets to their destinations, as long as the underlying network is connected~\cite{feigenbaum2012brief}.  
But also for multiple link failures, it is sometimes possible to achieve perfect resilience, for example on outerplanar graphs~\cite{FHPST21}.

It is also known that achievable resilience depends on the specific local information that the forwarding rules can rely on. Although it is typically assumed that the forwarding rules can depend on the status of the incident links, the packet's destination, and the link incident to the node from which the packet arrives (the so-called \emph{in-port}), the achievable resilience can increase if the rules can additionally also observe the packet's source. In particular, Dai et al.~\cite{dai2023tight} showed that it is always possible to tolerate two link failures if the packet's source can be taken into account. 

However, today, we generally still do not have a good understanding of the scenarios in which perfect resilience can be achieved.

This paper initiates the study of the computational complexity of configuring local fast re-routing algorithms. 
In particular, we investigate the question whether the resilience of a network can be verified efficiently:

\begin{itemize}
    \item \emph{Does a given network and its local fast re-route mechanism provide perfect resilience?} 
\end{itemize}

Furthermore, we are interested in scenarios where networks can be efficiently configured.

\begin{itemize}
    \item \emph{In which scenarios can perfectly resilience networks be configured efficiently?} 
\end{itemize}

Our main contribution in this paper is the proof that the first problem is \coNP-complete. On the positive side, we also show that in a subclass of practically relevant scenarios, the second problem can be solved efficiently.

In addition to perfect resilience, we also study a weaker notion of resilience often considered in the literature: \emph{ideal resilience} \cite{frr-ton}. Here it is assumed that initially (before the failures), the network is $k$-link-connected, that is, the network cannot be partitioned into isolated components with up to $k-1$ link failures. In such highly connected cases, the routing is called \emph{ideal} if it can tolerate up to $k-1$ link failures (which implies that the network remains connected after the failures). 
In this paper, we also show that given a network and its packet forwarding rules, it is \coNP-complete to decide whether the forwarding rules achieve ideal resilience.

\subsection{Contributions}
We show that in the standard scenario where nodes know the in-port (the link on which a packet arrived), verifying whether a given network and its configuration provide ideal resilience is \coNP-complete. This is even the case when the routing functions are restricted to the simple case of priority lists (``skipping''). For perfect resilience, we show that the same is true, and the hardness result even holds when we restrict ourselves to the class of planar graphs.

On the positive side, for scenarios where nodes forward packets independently of the link they received it on (i.e., the routing is ``in-port oblivious''), verifying whether a given routing is perfectly resilient, as well as deciding whether a given graph permits perfectly resilient routing, are both decidable in linear time.
The in-port oblivious routings are interesting as they allow us to save memory for storing the routing tables. 

\subsection{Additional Related Work}

Local fast re-routing mechanisms have been studied intensively in the literature already,
and we refer the reader to the recent survey by Chiesa et al.~\cite{frr-survey} for a detailed overview. There also exist industrial standards for resilient routing for most modern network protocols, from IP networks~\cite{ipfrr-survey-1-long} to MPLS networks~\cite{mplsfrr} to recent segment routing~\cite{bashandy2018topology} and software-defined networks~\cite{offrr}. However, these standards typically focus on single link failures and do not provide perfect or ideal resiliency guarantees. 

There is a large body of applied literature in the networking community on the topic~\cite{conext19failover,foerster2018ti,INFOCOM21-FFR,holterbach2017swift,jensen2020aalwines,frr-ankit}, typically focusing on heuristics.
In contrast, in our paper, we focus on algorithms which provide formal resilience guarantees. 

There are several interesting theoretical results for failure scenarios in which the number $f$ of link failures is bounded. 
In this context, a local fast rerouting scheme which tolerates~$f$~link failures is called \emph{$f$-resilient}. Chiesa et al.~\cite{frr-infocom16} showed that $2$-resilient routing is always possible if the graph is $3$-link-connected, and Dai et al.~\cite{dai2023tight} showed that $2$-resilience is always possible if rerouting rules can also depend on the source of a packet in addition to the target.
In their APOCS 2021 paper, Foerster et al.~\cite{FHPST21} proved that network topologies that form outerplanar graphs are always perfectly resilient and allow for simple and efficient rerouting algorithms based on skipping (each node stores an ordered priority list of alternative links to try per in-port and failed links are then simply skipped). While skipping leads to compact routing tables, it is an open question whether rerouting algorithms which are restricted to skipping come at a price of reduced resilience~\cite{frr-ton,FHPST21}.

Perfectly resilient network configurations, when they exist, can also be generated automatically with recent tools such as SyPer~\cite{gyorgyi2024syper} and SyRep~\cite{frr-syrep24}, using binary decision diagrams with quantification. This also implies that the complexity of the problem of generating such rules is in~\textsf{PSPACE}.

It is still an open research question whether ideal resilience can always be achieved, but it has been shown to be always achievable for all~$k\leq 5$~\cite{frr-ton}. It is also known that at least $\lfloor\frac{k}{2}\rfloor-1$ failures can be tolerated for general $k$~\cite{frr-infocom16}. It has further been proved that $k-1$ failures can be tolerated for special graphs including cliques, complete bipartite graphs, hypercubes, or Clos networks, as well as in scenarios where the rerouting rules can additionally also depend on the source~\cite{frr-ton}.
The corresponding algorithms are based on arborescence decompositions of the underlying graph, where all arborescences are rooted at the target.
If a packet traveling along an arborescence encounters a failed link, it is rerouted to the next arborescence in a circular order (known as circular arborescence routing). This technique builds upon Edmond's classic result on edge-disjoint branchings and arborescent decompositions~\cite{1973Edmonds}, and was used in many papers~\cite{frr-survey}.

Researchers have also already explored various variations for local fast rerouting models. In particular, while we in this paper focus on deterministic algorithms, there are also results on randomized algorithms: models where routers can generate random numbers or hash packet headers. Chiesa et al.~\cite{icalp16} in their ICALP 2016 paper considered $k$-link-connected graphs and showed that a simple randomized algorithm not only achieves a high robustness but also results in short paths (in the number of hops). Bankhamer et al.~\cite{disc21} in their DISC 2021 paper considered datacenter networks (based on Clos topologies) and showed that as long as the number of failures at each node does not exceed a certain bound, their algorithm
achieves an asymptotically minimal congestion up to polyloglog factors along failover paths. 

\subsection{Organization}

The remainder of this paper is organized as follows. 
We introduce our model and notations more formally in Section~\ref{sec:definitions}.
Our main contribution, the complexity results, appear in 
Section~\ref{sec:complexity}.
We conclude and provide open questions in 
Section~\ref{sec:conclusion}.

\section{Definitions}\label{sec:definitions}

We define a network as an undirected graph~$G=(V,E)$, where $V$ is a finite set of nodes and~${E \subseteq \binom{V}{2}}$ is a set of links.
Let $\local(v)$ be the set of links that are incident to a node~${v \in V}$.
To simulate the injection of a packet into the network, we moreover require the existence of the
self-loop~$\{v\} \in E$ for each node $v \in V$.
Since these self-loops always exist, we will not draw them in figures.

A \emph{path} in $G$ is a sequence~$(v_1,v_2, \ldots, v_n)$ of nodes such that~$\{v_i,v_{i+1}\} \in E$ for all $1 \leq i < n$.
Two nodes $u,v \in V$ are \emph{connected} if there is a path
$(u=v_1,v_2, \ldots, v_n=v)$ between~$u$ and~$v$ in~$G$.
In the rest of this paper, we shall consider only connected graphs, meaning
that there is a path between any two nodes.

\begin{definition}
A \emph{failure scenario} $F \subseteq E$ is a subset of links not including any self loops, that is, a set with~${F \cap \{\{v\} \mid v \in V\} = \emptyset}$.
\end{definition}

We assume that the self-loops are only used to model the arrival of packets and therefore do not allow them to fail.
Links in $F$ are referred to as \emph{failed} links and in~$E \setminus F$ as \emph{active} links.
Due to our assumption, every node has at least one active edge (self-loop) in any failure scenario.

\begin{definition}
A \emph{forwarding pattern} $\rho$ is a collection of functions~$\rho^F_v$ for each failure scenario~$F$ and each node~$v$, where each function~$\rho^F_v: \local(v) \rightarrow \local(v)$ returns an active link~$e' \in \local(v) \setminus F$ for every link~$e\in \local(v)$.
\end{definition}

From now on, we assume a fixed target node $t \in V$ such that all packets that arrive at any node in the graph should be forwarded to the node $t$. This corresponds to the (weakest possible) assumption of only destination matching when analyzing packet headers.

\begin{definition}
A forwarding pattern~$\rho$, a starting node~$s$, and a set~$F \subseteq E$ of failed links determine a \emph{routing}~$\pi=(\pi_1,\pi_2,\ldots)$ (a sequence of nodes) as follows: $\pi_1 = s$ and for each~$i \geq 1$, $\rho_{\pi_{i}}^F(\{\pi_{i-1},\pi_{i}\}) = \{\pi_{i},\pi_{i+1}\}$ (where~$\pi_0 = s$). The sequence stops at the lowest index $j$ where~$\pi_j = t$ (when the packet reaches the destination node $t$).
\end{definition}

We refer to \cref{fig:intro-ex-routing} for an example of a forwarding pattern and a routing.

\begin{definition}
Let $G=(V,E)$ be an undirected graph and let~$t \in V$ be a given target node.
A forwarding pattern~$\rho$ is \emph{perfectly resilient} if for every failure scenario~$F$ and every node~$v \in V$ that is connected to~$t$ via some path in~$G_F = (V,E\setminus F)$, the routing determined by~$\rho$,~$v$, and~$F$ is finite (contains~$t$).
\end{definition}
In other words, a routing $\rho$ is perfectly resilient if for any failure scenario $F$ and any node~$v$ that is connected to $t$ under $F$, a packet injected via the self-loop to the node~$v$ is eventually delivered to the target $t$.

A special case of the perfect resilience property that attracted a lot of attention in the literature is that of ideal resilience~\cite{frr-infocom16,Chiesa2014}.

\begin{definition}
Let $G=(V,E)$ be a $k$-link-connected graph (removing any set of~$k-1$ links does not disconnect~$G$).
A forwarding pattern~$\rho$ is \emph{ideally resilient} if for any failure scenario~$F$ of size at most~$k-1$ 
and every node $v \in V$, the routing determined by~$\rho$, $v$, and~$F$ is finite (contains $t$).
\end{definition}

Clearly, if the forwarding pattern~$\rho$ has complete knowledge of the global failure scenario, it is always possible to construct a perfectly/ideally 
resilient forwarding pattern. As we are interested in forwarding pattern that allow for fast re-routing based only on the information about locally failing links, we define the notion of a local forwarding pattern.

\begin{definition}
A forwarding pattern~$\rho$ is \emph{local} if $\rho_v^F(e) = \rho_v^{F'}(e)$
for every~$v \in V$, every~$e \in E$, and any two failure scenarios~$F$ and~$F'$ such that~$F \cap \local(v) = F' \cap \local(v)$. 
\end{definition}

Hence, the decision of the next-hop can only depend on the knowledge of active and failing links incident to a given node.
We shall also call such a general local forwarding pattern \emph{combinatorial} because the size of a local routing table for a node $v$ requires $2^{|\local(v)|}$ routing entries and hence grows exponentially with the degree of $v$ (the size of $\local(v)$). Storing such forwarding patterns is impractical due to high memory requirements.
In order to save memory in routing tables, several works consider the more practical
notion of skipping forwarding patterns~\cite{frr-infocom16,frr-ton,FHPST21}.

\begin{definition}
A \emph{skipping priority list} is a function~${\rho_v: \local(v) \rightarrow \local(v)^*}$ such that for every link~$e \in \local(v)$ the function $\rho_v(e)=(e_1,e_2, \ldots, e_n)$ returns a permutation of all links in $\local(v)$. A given skipping priority list $\rho$ determines a \emph{skipping forwarding pattern} $\rho$ such that
$\rho_v^F(e)=e_i$ where $\pi(e)=(e_1,e_2,\ldots, e_n)$ and $i$ is the lowest
index such that $e_i \notin F$.
\end{definition}

\begin{obs}
	Clearly, any skipping forwarding pattern is also local (combinatorial) because the forwarding decision is based only on the status of incident links. Moreover, representing a skipping priority list~$\rho_v(e)$ requires us to store only $|\local(v)|$ links.
\end{obs}

We refer to \cref{fig:intro-ex-routing} for an example of a graph with a perfectly resilient skipping forwarding pattern.
We depict the priority list function as a routing table
where the first column is the node that receives a packet on the given in-port (in the second column) and returns
the respective priority list in the third column. For example, the first row in the routing table denotes the routing
entry ${\rho_{v_1}(\{v_1,v_1\})=(\{v_1,v_2\},\{v_1,v_3\},\{v_1,v_4\},\{v_1,v_1\})}$. The example shows the routing for a packet injected at the node $v_1$ in the scenario where no links fail and when the links~$\{v_2,v_5\}$ and~$\{v_3,v_5\}$ fail. As the forwarding pattern is perfectly resilient, in both scenarios the packet is delivered to the target node $t=v_5$.

\begin{figure}[t]
    \subcaptionbox{Example graph with the target node~$t=v_5$.}[0.21\textwidth]{
	    \centering
    	\begin{tikzpicture}[baseline,remember picture,scale=1,every node/.style={scale=0.9}, >=stealth', shorten <= 1pt, shorten >= 1pt]
			\foreach \x / \y / \n in
			{0/0/p1, -1/-1/p2, 0/-1/p3, 1/-1/p4, 0/-2/p5}
			{
				\node[nodeA] at (1.2*\x, 2*\y) (\n) {};
			}
			\foreach \n in {1,2,3,4,5} {
				\node[right = 0 pt of p\n] {$v_\n$}; 
			}
			
			\foreach \v / \w in {1/2, 1/3, 1/4, 2/5, 3/5, 4/5} {
				\draw[thick] (p\v) edge (p\w);
			}
            \draw[red,thick] (p2) edge (p5);
            \draw[red,thick] (p3) edge (p5);
		\end{tikzpicture}
        }
		\label{subfig:intro-ex-graph}
    \hspace{.3cm}
    \subcaptionbox{Perfectly resilient skipping forwarding pattern; a routing with no failing edges when a packet is injected to $v_1$ is $(v_1,v_2,v_5)$.}[0.35\textwidth]{
        \centering
    	\begin{tikzpicture}[baseline,remember picture,scale=1,every node/.style={scale=0.9}, >=stealth', shorten <= 1pt, shorten >= 1pt]
			\node (tab1) at (0, 0) {%
				\renewcommand{\arraystretch}{1.0}
				\begin{tabular}{c|c|r}
					node & in-port & priority list \\
					\hline 
					$v_1$ & \colorbox{yellow}{$v_1$} & \colorbox{yellow}{$v_2$}, $v_3$, $v_4$, $v_1$ \\
                    $v_1$ & $v_2$ &$v_3$, $v_4$, $v_2$, $v_1$ \\
                    $v_1$ & $v_3$ &$v_4$, $v_2$, $v_3$, $v_1$ \\
                    $v_1$ & $v_4$ &$v_2$, $v_3$,  $v_4$, $v_1$ \\
                    \hline
					$v_2$ & \colorbox{yellow}{$v_1$} & \colorbox{yellow}{$v_5$}, $v_1$, $v_2$ \\
                    $v_2$ & $v_2$ &$v_5$, $v_1$, $v_2$ \\
                    \hline
					$v_3$ & $v_1$ & $v_5$, $v_1$, $v_3$ \\
                    $v_3$ & $v_3$ &$v_5$, $v_1$, $v_3$ \\
                    \hline 
					$v_4$ & $v_1$  &$v_5$, $v_1$, $v_4$ \\
					$v_4$ & $v_4$ &$v_5$, $v_1$, $v_4$
			\end{tabular}};
		\end{tikzpicture}
		\label{subfig:intro-ex-table}
    }
    \hspace{.3cm}
    \subcaptionbox{A routing for the failure scenario $F=\{\{v_2,v_5\},\{v_3,v_5\}\}$ is~$(v_1,v_2,v_1,v_3,v_1,v_4,v_5)$.}[0.35\textwidth]{
        \centering
		\begin{tikzpicture}
		[baseline,remember picture,scale=1,every node/.style={scale=0.9}, >=stealth', shorten <= 1pt, shorten >= 1pt]
			\node (tab1) at (0, 0) {%
				\begin{tabular}{c|c|r}
					node & in-port & priority list \\
					\hline 
					{$v_1$} & \colorbox{yellow}{$v_1$} & \colorbox{yellow}{$v_2$}, $v_3$, $v_4$, $v_1$ \\
                    $v_1$ & \colorbox{yellow}{$v_2$} &\colorbox{yellow}{$v_3$}, $v_4$, $v_2$, $v_1$ \\
                    $v_1$ & \colorbox{yellow}{$v_3$} & \colorbox{yellow}{$v_4$}, $v_2$, $v_3$, $v_1$ \\
                    $v_1$ & $v_4$ &$v_2$, $v_3$,  $v_4$, $v_1$ \\
                    \hline
					$v_2$ & \colorbox{yellow}{$v_1$} & \cancel{$v_5$}, \colorbox{yellow}{$v_1$}, $v_2$ \\
                    $v_2$ & $v_2$ &\cancel{$v_5$}, $v_1$, $v_2$ \\
                    \hline
					$v_3$ & \colorbox{yellow}{$v_1$} & \cancel{$v_5$}, \colorbox{yellow}{$v_1$}, $v_3$ \\
                    $v_3$ & $v_3$ &\cancel{$v_5$}, $v_1$, $v_3$ \\
                    \hline 
					$v_4$ & \colorbox{yellow}{$v_1$}  &\colorbox{yellow}{$v_5$}, $v_1$, $v_4$ \\
					$v_4$ & $v_4$ &$v_5$, $v_1$, $v_4$
			\end{tabular}};
		\end{tikzpicture}
        \label{subfig:intro-ex-routing}
    }
	\caption{Example of a graph with target $t=v_5$ and a perfectly resilient skipping forwarding pattern. For the sake of readability, we represent (here and in the following figures) links incident to a node~$v$ by only stating the other endpoint of the link. The routing determined by this forwarding pattern, failure scenario~$F=\{\{v_2,v_5\},\{v_3,v_5\}\}$, and the starting node~$s=v_1$ is~${\pi=(v_1,v_2,v_1,v_3,v_1,v_4,v_5)}$.}
	\label{fig:intro-ex-routing}
\end{figure}

A more restricted type of routing is the one that does not consider the in-port information for the 
routing decision. This is formalized in the following definition.

\begin{definition}
A forwarding pattern~$\rho$ is \emph{in-port oblivious}
if $\rho_v^F(e)=\rho_v^F(e')$ for any node~$v \in V$, any failure scenario $F$, and any two links~$e, e' \in \local(v)$. 
\end{definition}

Hence, when using in-port oblivious forwarding patterns, a~node~$v$ makes always the same next-hop, regardless of at which link (port) the packet arrives at $v$.
In-port oblivious forwarding patterns are more memory efficient as they requires fewer entries in the routing tables.

Finally, we shall introduce two decision problems related to perfect and ideal 
resilience.
In the \emph{verification problem}, we are given an undirected graph $G=(V,E)$, a target node~${t \in V}$, and a local/skipping forwarding pattern~$\rho$.
The task is to decide whether~$\rho$ is perfectly/ideally resilient.
In the \emph{synthesis problem}, we ask whether there exists a local (or skipping) forwarding pattern that is perfectly/ideally resilient for a given graph~$G$ and a target node~$t$.

\section{Complexity Results}\label{sec:complexity}
We shall start by stating an obvious upper-bound on the complexity of perfect resilience verification
and at the same time point out to an interesting fact
that the synthesis problem is decidable in polynomial time
thanks to the well-know Robertson and Seymour~\cite{RS-theorem} theorem.

\begin{theorem} \label{thm:complexity}
Verification of \textsc{Perfect Resilience} and verification of \textsc{Ideal Resilience} are in~\coNP. Synthesis of \textsc{Perfect Resilience} is in~\poly.
\end{theorem}
\begin{proof}
    For the verification problems, we can guess a failure scenario and a source node that is connected to the target but where the forwarding pattern creates a forwarding loop, which can be checked in polynomial time for the given failure scenario. Hence, if there exists a failure scenario creating a forwarding loop, the given instance of the verification problem is not perfectly resilient, showing that the problem belongs to \coNP.

    The containment in \poly{} for the synthesis problem is a direct consequence of the fact that perfectly resilient networks are closed under taking minors~\cite{FHPST22} and the seminal result of Robertson and Seymour~\cite{RS-theorem} showing that any such property can be characterized by a finite set of forbidden minors.
    Combined with a~polynomial time algorithm for checking whether a~graph contains a~forbidden minor, we know that there must exist an algorithm with a polynomial running time that decides perfect resilience for the case where all nodes can become possible target nodes.
    Using rooted minors (the target~$t$ is a root), it is easy to show that deciding whether a~given graph with a fixed given target node is perfectly resilient can be also done in polynomial time~\cite{RS90,RS95}.
\end{proof}

However, the line of arguments used in the proof to show that the verification problem of perfect resilience is polynomial-time solvable does not provide a concrete algorithm as it does not tell us the set of forbidden minors; it only points to its existence.
The problem of effectively (in polynomial time) designing perfectly resilient 
forwarding patterns, in case they exist, remains a major open problem.

Next, we shall study the complexity of in-port oblivious perfect resilience, which is an interesting problem from the practical point of view as in-port oblivious forwarding patterns can be stored with reduced memory foot-print.

\subsection{Routings with no In-port Matching}
We shall first state a necessary and sufficient condition for a graph to allow for in-port oblivious forwarding patterns.

  \begin{lemma} \label{lem:inport}
    Let $G=(V,E)$ be a connected graph.
    There exists a perfectly resilient in-port oblivious local (or skipping) forwarding pattern for any given target node~$t$ if and only if all simple cycles of $G$ have length at most $3$.
  \end{lemma}
  
\begin{proof}
For the forward direction, assume towards a contradiction that there is an in-port oblivious perfectly resilient local forwarding pattern~$\rho$ for a given target node~$t$ and that the graph $G$ contains a simple cycle of length $4$ or more. The existence of this cycle implies
that there are three different nodes $s,u,v \in V \setminus \{t\}$ such that
        \begin{itemize}
        \item $\{s,u\} \in E$ and $\{s,v\} \in E$, 
        \item there is a path from $u$ to $t$ which does not contain the nodes~$v$ or~$s$, and
        \item there is a path from $v$ to $t$ which does not contain the nodes~$u$ or~$s$.
    \end{itemize}
    Now a packet injected to $s$ in a failure scenario where
    all links incident to~$s$ fail, except for $\{s,u\}$ and $\{s,v\}$, will have to be forwarded to either $u$ or $v$. Let us assume without loss of generality that it is forwarded to $u$. We extend the failure
    scenario by failing also all links incident to~$u$, except for the link between $u$ and $s$. This means that the packet
    returns back to~$s$ and because the forwarding pattern is local and cannot see the failed links at the node~$u$, and we are
    in-port oblivious,
    the node $s$ has to forward the packet back to~$u$. A forwarding loop
    is hence created; however, the node $s$ is still connected to the target $t$ via the node $v$. This is a contradiction to the assumption that~$\rho$ is a perfectly resilient forwarding pattern.

     \begin{figure}[t]
         \centering
     \begin{tikzpicture}[scale=0.90,every node/.style={scale=0.9}, >=stealth', shorten <= 1pt, shorten >= 1pt]
         \foreach \x / \y / \n in
         {1/1/u, 2/2/v, 1/0/u', 1/-1/t}
         {
             \node[nodeA] at (\x, \y) (\n) {};
             \node[right = 9pt and -3pt of \n] {\,$\n$}; 
         }
          \foreach \x / \y / \n in
         { 0/2/s}
         {
             \node[nodeA] at (\x, \y) (\n) {};
             \node[left = 9pt and -3pt of \n] {\,$\n$}; 
         }
         \draw (u) edge (u');
         \draw (t) edge[dotted] (u');
         \draw (s) edge (u);
         \draw (s) edge (v);
         \draw (v) edge (u);
     \end{tikzpicture}		
     \caption{A triangle consisting of three nodes~$s,v,$ and~$u$, where~$u$ is the unique node where all paths from either~$s$ or~$v$ to~$t$ pass through~$u$ and~$u'$ is the second node on the shortest path from~$u$ to~$t$.}
     \label{fig:triangle}
     \end{figure}
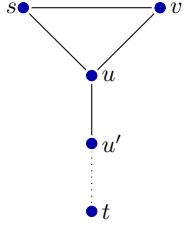

In the backward direction, assume that the graph~$G$ contains only cycles of length at most $3$. This implies that from any node in~$G$ there is a unique shortest path to the given target node~$t$. We construct an in-port oblivious perfectly resilient skipping (and hence also local) forwarding pattern~$\rho$ as follows. For every node~$s$ that is not on any simple cycle of length $3$ in $G$, we define~${\rho(\{*,s\})=\{s, u\}}$ where~$u$ is the first node on the unique shortest path from~$s$ to~$t$ and $*$ stands for any neighbor of~$s$. There is no point in sending the packet along any other outgoing link as it will have to necessarily return back to~$s$ and hence create a forwarding loop. Let us now consider three nodes $s, u, v \in V$ that are on a cycle of length~$3$ as depicted in Figure~\ref{fig:triangle}. 
Let~$u$ be the unique node that has a path to the target~$t$ without visiting~$s$ or~$v$. We define $\rho(\{*,u\})=\{u,u'\}$ where~$u'$ is the unique next-hop node on the shortest path from $u$ to $t$. As before, this is the only choice should $\rho$ be perfectly resilient and the links following after $\{u,u'\}$ in the priority list are irrelevant.
Finally, we set~$\rho(\{*,s\})=\{s,u\}\{s,v\}$ and~$\rho(\{*,v\})=\{v,u\}\{v,s\}$.
Note that whenever there is a path from~$s$ to~$t$, then each such path contains~$u$.
Hence, if the link~$\{s,u\}$ does not fail, then forwarding to~$u$ always makes progress towards reaching~$t$.
If the link fails, then all paths from~$s$ to~$t$ contain both~$v$ and~$u$ and hence forwarding to~$v$ also makes progress.
The same argument applies to paths starting in~$v$ and thus the constructed forwarding pattern is perfectly resilient, concluding the proof.
\end{proof}

\begin{figure}[t]
    \begin{center}
    \begin{subfigure}[c]{0.15\textwidth}
    \begin{tikzpicture}[scale=0.95,every node/.style={scale=0.9}, >=stealth', shorten <= 1pt, shorten >= 1pt]
				\foreach \x / \y / \n in
				{0/0/v_1, 1/0/v_6, 0/1/v_2, 1/2/v_3, 1/1/v_5, 2/1/v_4, 2/0/t}
				{
					\node[nodeA] at (\x, \y) (q\x\y) {};
                    \node[above left = -3pt and -3pt of q\x\y] {$\n$}; 
				}
                \foreach \v / \w in {00/11,00/01,01/11,12/11,21/11,21/10,11/10,10/20} {
					\draw (q\v) edge (q\w);
				}
			\end{tikzpicture}		
            \end{subfigure}\hspace{2cm}
        \begin{subfigure}{0.27\textwidth}     
            \begin{tabular}{c|r}
					node & priority list \\
					\hline 
					$v_1$ & \colorbox{yellow}{$v_5$}, \colorbox{yellow}{$v_2$}, $v_1$ \\
                    $v_2$ & \colorbox{yellow}{$v_5$}, \colorbox{yellow}{$v_1$}, $v_2$ \\
                    $v_3$ & \colorbox{yellow}{$v_5$}, $v_3$ \\
                    $v_4$ & \colorbox{yellow}{$v_6$}, \colorbox{yellow}{$v_5$}, $v_4$ \\
                    $v_5$ & \colorbox{yellow}{$v_6$}, \colorbox{yellow}{$v_4$}, $v_1$, $v_2$, $v_3$,$v_5$ \\
                    $v_6$ & \colorbox{yellow}{$t$}, $v_4$, $v_5$, $v_6$
                    \end{tabular}
                    \end{subfigure}
    \caption{A graph and a perfectly resilient in-port oblivious skipping forwarding pattern for the target $t$.}
    \label{fig:inport-ex}
    \end{center}
\end{figure}
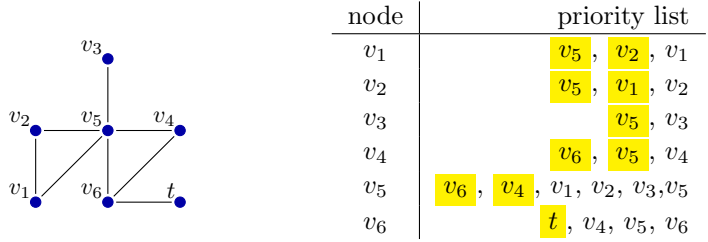

An example of a graph with a target node~$t$ allowing for a perfectly resilient skipping forwarding pattern is shown in~\cref{fig:inport-ex}.
Only the highlighted entries are important.
We can now conclude with the fact that both the verification and synthesis problems for in-port oblivious routings are efficiently solvable.
  
\begin{theorem}
   The verification and synthesis problems for \textsc{Perfect Resilience} using in-port oblivious local and skipping 
   routings are decidable in linear time.
\end{theorem}
\begin{proof}
   First, by using Lemma~\ref{lem:inport}, we shall check
   for the existence of cycles of length $4$ or more in the connected component containing~$t$.
   Such a check can be performed in linear time using, e.g., a slightly modified BFS from $t$ where for each node $v$, we keep track of the parent $p(v)$ and for each link not part of the BFS tree, we check that the two endpoints have the same parent and each node is only incident to one such link.
    If this linear-time check fails, we
    output that no perfectly resilient local forwarding pattern exists. Otherwise, by Lemma~\ref{lem:inport} there is an in-port oblivious skipping forwarding pattern~$\rho$ and we have a yes-instance for the synthesis problem.
    Algorithm~\ref{alg:inportoblivious-synthesis} is a pseudocode implementation of this algorithm.
    
    For the verification problem, we just have to check that the important priority list entries for the given skipping or combinatorial routing agree with the choices made by $\rho$ in any local failure scenario.
   If they agree with the entries computed by the procedure described
    in Lemma~\ref{lem:inport} then the given forwarding pattern is perfectly resilient.
    
    It remains to show that whenever the forwarding pattern does not agree with the forwarding pattern described in \cref{lem:inport} in important entries, then it is not perfectly resilient.
    To this end, observe that each important entry in the forwarding pattern for a node~$v$ is a node on a path from~$v$ to~$t$.
    Consider a failure scenario~$F_1$ where the assumed forwarding pattern chooses a different entry than the one behind \cref{lem:inport}.
    Let~$u$ be the one chosen by~\cref{lem:inport} and~$u'$ be the one chosen by the assumed solution.
    Note that~$u' \neq t$.
    Consider a failure scenario where all links in~$F_1$ incident to~$v$ fail and all links incident to~$u'$ except for~$\{v,u'\}$.
    Then, the forwarding pattern has to return the package to the initial node~$v$ and hence we obtain a forwarding loop.
    However, there is still a path from~$v$ to~$t$ via~$u$ which consists only of active links, a contradiction to the assumed forwarding pattern to achieve perfect resilience.
\end{proof}

\begin{algorithm}[t]
    \caption{Tests the existence of an in-port oblivious perfectly resilient forwarding pattern. Takes a graph $G=(V,E)$ and a target node~$t$ as input and outputs $\text{true}$ if such a routing exists; $\text{false}$ otherwise.}
    \label{alg:inportoblivious-synthesis}
    \begin{algorithmic}[1]
        \Procedure{Synthesis}{$V, E, t$}
            \State Initialize queue $Q \gets [t]$
            \ForAll{$v \in V$}
                \State $visited[v] \gets \text{false}$
                \State $parent[v] \gets \text{null}$
            \EndFor
            \State $visited[t] \gets \text{true}$
            \While{$Q \neq \emptyset$}
                \State $u \gets$ \Call{Dequeue}{$Q$}
                \State $cycleFound \gets \text{false}$
                \ForAll{$\{u,v\} \in E$}
                    \If{not $visited[v]$}
                        \State $visited[v] \gets \text{true}$
                        \State $parent[v] \gets u$
                        \State \Call{Queue}{$Q, v$}
                    \Else \Comment{$\{u, v\}$ is an off-tree edge}
                        \If{$parent[u] \neq parent[v]$}
                            \State \textbf{return} \text{false}
                        \EndIf
                        \If{$cycleFound$} \Comment{$u$ has more than one off-tree edge}
                            \State \textbf{return} \text{false}
                        \EndIf
                        \State $cycleFound \gets \text{true}$
                    \EndIf
                \EndFor
            \EndWhile
            \State \textbf{return} \text{true}
        \EndProcedure
    \end{algorithmic}
\end{algorithm}

\subsection{Hardness of Perfect Resilience Verification}

In this section, we show that the verification problem for perfect resilience  is \coNP-complete.
This results contrasts the fact that the seemingly more
difficult synthesis problem is decidable in polynomial time (see Theorem~\ref{thm:complexity}).

\begin{theorem}
    \label{thm:perfecthard}
    Verification for \textsc{Perfect Resilience} with skipping priorities is \coNP-complete even if the input graph is planar.
\end{theorem}

\begin{proof}
    Containment in \coNP is proved in Theorem~\ref{thm:complexity}.
    To show \coNP-hardness, we present a reduction from \textsc{3-Sat}, a variant of \textsc{Satisfiability} where each clause contains exactly~$3$ literals.
    This problem is among Karp's original 21 \NP-complete problems~\cite{Karp72}.
    We show a polynomial-time reduction from \textsc{3-Sat} to the verification problem of \textsc{Perfect Resilience} where the constructed skipping routing is perfectly resilient if and only if the original instance of \textsc{3-Sat} (a formula~$\phi$) is \emph{not} satisfiable.
    
    Let~$\mathcal{V}=\{x_1,x_2,\ldots,x_{n}\}$ and~${\mathcal{C}=\{C_1,C_2,\ldots,C_{m}\}}$ be the set of variables and clauses of an input formula~$\phi$.
    We construct an instance of the verification problem for \textsc{Perfect Resilience} with skipping priorities.
    We start with two nodes~$c$ and~$t$, which are connected by a link and where~$t$ is the target.
    Next, for each variable~${x_i \in \mathcal{V}}$, we add four nodes~$p_{x_i}, n_{x_i}, p_{\overline{x_i}}$, and~$n_{\overline{x_i}}$ and the six links~$\{c,p_{x_i}\}, \{c,n_{x_i}\}, \{p_{x_i},n_{x_i}\}, \{c,p_{\overline{x_i}}\}, \{c,n_{\overline{x_i}}\}$, and~$\{p_{\overline{x_i}},n_{\overline{x_i}}\}$.
    For each clause~$C_j \in \mathcal{C}$, we add three nodes~$u_j,v_j,w_j$ and three links~$\{c,u_j\},\{c,v_j\},\{c,w_j\}$.
    See \cref{fig:network1} for an example of the construction and note that the graph is planar and the number of nodes and links is linear in the number of variables and clauses in~$\phi$.
    We continue with the priority lists.
    Note that if the link~$\{c,t\}$ fails, then only packets starting in~$t$ can reach~$t$.
    Hence, we assume in the remainder of the proof that this link does not fail.
    Hence, we only need to state the priority lists for~$c$ (for different in-ports) up to~$t$.
    The priority lists are given in \cref{fig:table1}.
    \begin{figure}[t]
        \centering
        \begin{subfigure}[t]{0.4\textwidth}
        \centering
        \vspace{-10cm}
        \begin{tikzpicture}
            \node[circle,draw,inner sep=3pt] (c) at(0,0) {$c$};
            \foreach \i/\j in {1/$p_{x_1}$,2/$n_{x_1}$,3/$p_{\overline{x_1}}$,4/$n_{\overline{x_1}}$,5/$p_{x_2}$,6/$n_{x_2}$,7/$p_{\overline{x_2}}$,8/$n_{\overline{x_2}}$,9/$p_{x_3}$, 10/$n_{x_3}$,11/$p_{\overline{x_3}}$,12/$n_{\overline{x_3}}$,13/$u_1$,14/$v_1$,15/$w_1$,16/$u_2$,17/$v_2$,18/$w_2$,19/$t$}{
                \node[circle,draw] at(90-\i*360/19:1.5) (v\i) {};
                \node at(90-\i*360/19:2) {\j};
            }
            \foreach \i in {1,2,...,19}{
                \draw[-] (v\i) to (c);
            }
            \draw[-] (v1) to (v2);
            \draw[-] (v3) to (v4);
            \draw[-] (v5) to (v6);
            \draw[-] (v7) to (v8);
            \draw[-] (v9) to (v10);
        \end{tikzpicture}
        \caption{Construction of a graph for a formula~$\phi$ with three variables and two clauses.}
    \label{fig:network1}
\end{subfigure}\hspace{1cm}
\begin{subfigure}{0.4\textwidth}
        \begin{tabular}{c|c|r}
             node & in-port & priority list \\\hline
             $p_{x_i}$ & $c$ & $n_{x_i}, c$\\
             $p_{x_i}$ & $n_{x_i},p_{x_i}$ & $c, n_{x_i}$\\
             $n_{x_i}$ & $c$ & $p_{x_i}, c$\\
             $n_{x_i}$ & $p_{x_i},n_{x_i}$ & $c, p_{x_i}$\\
             $p_{\overline{x_i}}$ & $c$ & $n_{\overline{x_i}}, c$\\
             $p_{\overline{x_i}}$ & $n_{\overline{x_i}},p_{\overline{x_i}}$ & $c, n_{\overline{x_i}}$\\
             $n_{\overline{x_i}}$ & $c$ & $p_{\overline{x_i}}, c$\\
             $n_{\overline{x_i}}$ & $p_{\overline{x_i}},n_{\overline{x_i}}$ & $c, p_{\overline{x_i}}$\\\hline
             $u_j$ & $*$ & $c$\\
             $v_j$ & $*$ & $c$\\
             $w_j$ & $*$ & $c$\\\hline
             $c$ & $c$ & $p_{x_1}, p_{\overline{x_1}}, t$\\
             $c$ & $p_{x_i},p_{\overline{x_i}}$ & $t$\\
             $c$ & $n_{x_i}$ & $n_{\overline{x_i}}, p_{x_{i+1}}, p_{\overline{x_{i+1}}}, t$\\
             $c$ & $n_{\overline{x_i}}$ & $n_{x_i}, p_{x_{i+1}}, p_{\overline{x_{i+1}}}, t$\\
             $c$ & $n_{x_{n}}$ & $n_{\overline{x_{n}}}, u_1, v_1, w_1, t$\\
             $c$ & $n_{\overline{x_{n}}}$ & $n_{x_{n}}, u_1, v_1, w_1, t$\\
             $c$ & $u_m$ & $n_{q_m}, p_{x_1}, p_{\overline{x_1}}, t$\\
             $c$ & $v_m$ & $n_{r_m}, p_{x_1}, p_{\overline{x_1}}, t$\\
             $c$ & $w_m$ & $n_{s_m}, p_{x_1}, p_{\overline{x_1}}, t$\\
             $c$ & $u_j$ & $n_{q_j}, u_{j+1}, v_{j+1}, w_{j+1}, t$\\
             $c$ & $v_j$ & $n_{r_j}, u_{j+1}, v_{j+1}, w_{j+1}, t$\\
             $c$ & $w_j$ & $n_{s_j}, u_{j+1}, v_{j+1}, w_{j+1}, t$
        \end{tabular}
        \caption{Table of the priority lists for all nodes (except for the target~$t$). The symbol~$*$ stands for the case that the priority list does not depend on the in-port. We use~$q_j,r_j,$ and~$s_j$ to denote the three literals in clause~$C_j$, respectively. As an example, if clause~$C_0 = (x \lor y \lor \overline{z})$, then~$q_0 = x$, $r_0 = y$, and~$s_i = \overline{z}$. Priority lists for node~$c$ are only shown up to the target~$t$.}
    \label{fig:table1}
\end{subfigure}
        \caption{Example of the construction of an instance in the proof of \cref{thm:perfecthard}.}
        \label{tab:prio}
    \end{figure}
    Before we formally state the proof, let us give some high-level intuition.
    The goal is to find a set of failing links that cause the packet to loop forever between~$c$ and variable nodes~($p_{x_i},n_{x_i},p_{\overline{x_i}}$ and~$n_{\overline{x_i}}$) and clause nodes~($u_j,v_j$, and~$w_j$).
    First, if the center node~$c$ ever receives the packet from a node~$p_{x_i}$ or~$p_{\overline{x_i}}$ for some~$x_i$, then it immediately forwards the packet to~$t$.
    We next show for each~$i$ that exactly one of the links~$\{c,n_{x_i}\}$ and~$\{c,n_{\overline{x_i}}\}$ fails in any solution.
    Assume that~$c$ receives the packet from~$n_{x_{i-1}}$ or~$n_{\overline{x_{i-1}}}$ (or from~$u_{m},v_m,w_m$ or it starts in~$c$ for~$i=1$).
    Then not both the links~$\{c,p_{x_i}\}$ and~$\{c,p_{\overline{x_i}}\}$ can fail as otherwise the packet is sent to~$t$.
    Assume the packet is sent to~$p_{\overline{x_i}}$ (the other case is symmetric).
    If either of the links~$\{p_{\overline{x_i}},n_{\overline{x_i}}\}$ or~$\{n_{\overline{x_i}},c\}$ fails, then the packet is sent back from~$p_{\overline{x_i}}$ to~$c$ and then sent to~$t$.
    Hence, neither of these links fail and the packet is forwarded over these links.
    Now assume the link~$\{c,n_{x_i}\}$ does not fail.
    Then, the packet is sent through that link.
    Now at least one of the links~$\{n_{x_i},p_{x_i}\}$ and~$\{p_{x_i},c\}$ has to fail or~$c$ receives the packet from~$p_{x_i}$.
    Hence, the packet is sent back to~$c$ from~$n_{x_i}$.
    However, by construction the packet is now sent to~$n_{\overline{x_i}}$ and consequently to~$p_{\overline{x_i}}$, back to~$c$, and then to~$t$.
    Thus, exactly one of the links~$\{c,n_{x_i}\}$ and~$\{c,n_{\overline{x_i}}\}$ fails in any solution and whichever does not fail, the entire respective loop also does not fail.
    We say that if the link~$\{c,n_{x_i}\}$ fails, then~$x_i$ is set to true and if the link~$\{c,n_{\overline{x_i}}\}$ fails, then~$x_i$ is set to false.

    We now show that the reduction is correct, that is, $\phi$ is satisfiable if and only if the constructed routing is not perfectly resilient.
    In the forward direction, assume that a satisfying assignment~$\beta$ for~$\phi$ exists.
    We let the links~$\{c,p_{x_i}\}$ and~$\{c,n_{x_i}\}$ fail for each variable~$x_i$ set to true by~$\beta$ and we let the links~$\{c,p_{\overline{x_i}}\}$ and~$\{c,n_{\overline{x_i}}\}$ fail whenever~$x_i$ is set to false by~$\beta$.
    Moreover, for each clause~${C_j = (q_j \lor r_j \lor s_j)}$, we pick one variable that satisfies the clause under~$\beta$.
    If~$q_j$ is picked, then we let links~$\{c,v_j\}$ and~$\{c,w_j\}$ fail.
    If~$r_j$ is picked, then we let links~$\{c,u_j\}$ and~$\{c,w_j\}$ fail.
    If~$s_j$ is picked, then we let links~$\{c,u_j\}$ and~$\{c,v_j\}$ fail.
    Now, the packet goes through~$p_{x_i}$ and~$n_{x_i}$ or~$p_{\overline{x_i}}$ and~$n_{\overline{x_i}}$ for all~$i$ in increasing order.
    Then, for each clause~$C_j$ in increasing order, it goes to the previously chosen node ($u_j,v_j$, or~$w_j$) and immediately returns to~$c$.
    Then, it goes to the next clause as we have a satisfying assignment and hence the link to~$n_{q_j}/n_{r_j}/n_{s_j}$ fails.
    Afterwards it again goes through all variables and this cycle continues indefinitely.
    Thus, the constructed routing is not perfectly resilient.

    In the other direction, assume that the constructed routing is not perfectly resilient and let~$F$ be a failure scenario where the packet does not reach~$t$ from some source node.
    We have already shown that exactly one of the links~$\{c,n_{x_i}\}$ and~$\{c,n_{\overline{x_i}}\}$ fails for each variable~$x_i$ and how this corresponds to an assignment.
    If the packet reaches some clause node~$u_j,v_j$, or~$w_j$ (or starts there and there exists some path to~$c$ and~$t$), then it is sent to~$c$.
    Then, $c$ tries to send the packet to a node~$n_{x_i}$ or~$n_{\overline{x_i}}$ for some~$x_i$.
    If this link does not fail (which corresponds to the case where the clause~$C_j$ is not satisfied by the respective assignment of the corresponding variable), then the packet is sent to the respective node~$n_y$ and subsequently to~$p_y$ and back to~$c$ from~$p_y$ as the entire loop is not failing whenever the link~$\{c,n_y\}$ is not failing (as shown above).
    If all three links to a clause gadget fail, then the packet is directly sent to~$t$, so we may assume that at least one of them is not failing.
    The first one that is not failing will encode the satisfying assignment for the respective clause.
    The packet is then sent to one of the three nodes for the next clause and the cycle continues.
    This cycle can only continue indefinitely, if all clauses are satisfied by the chosen assignment, that is, the formula~$\phi$ is satisfiable.
    This concludes the proof. We note that the hardness result holds even in case that we know
    the source of the packet as, e.g., injecting the packet at $c$ is already enough to argue
    for \coNP hardness.
\end{proof}

\subsection{Hardness of Ideal Resilience Verification}

We next turn towards ideal resilience.
We also show that the problem is \coNP-complete, however this time we cannot show hardness for the special case of planar graphs.
Note that any~$6$-connected graph has no nodes of degree less than six.
Hence, no such graph is planar as planar graphs are~$5$-degenerate, meaning that each subgraph (including the entire graph) has a node of degree at most~$5$.
Since the verification problem can be solved in polynomial-time for each constant~$k$, the verification problem for \textsc{Ideal Resilience} actually becomes polynomial-time solvable if the input graph is planar.

\begin{theorem}
    \label{thm:idealhard}
    Verification for \textsc{Ideal Resilience} with skipping priorities is \coNP-complete.
\end{theorem}

\begin{proof}
    As containment in \coNP is stated in Theorem~\ref{thm:complexity}, we focus on the hardness part. 
    We again present a reduction from \textsc{3-Sat} where the number of nodes is linear in the number of variables and clauses in the input formula and where the constructed forwarding pattern is ideally resilient if and only if the input formula is not satisfiable.
     
    Let~$\phi$ be a formula with variables~${\mathcal{V}=\{x_1,x_2,\ldots,x_n\}}$ and clauses~$\mathcal{C}=\{C_1,C_2,\ldots,C_m\}$.
    We start with a clique~$K$ of~$2n+2m+1$ nodes~$u_0,u_1,\ldots,u_{2n+2m}$.
    All of the following nodes will be connected to all nodes in~$K$.
    Next, we add the target~$t$, three nodes~$v_i,x_i$ and~$\overline{x_i}$ for each variable~$x_i \in \mathcal{V}$, and two node~$c_j$ and~$d_j$ for each clause~$C_j \in \mathcal{C}$.
    We also add a node~$v_{n+1}$.
    Apart from the links to~$K$, we add the following links for each variable~$x_i$:
    $\{v_i,x_i\},\{v_i,\overline{x}_i\},\{x_i,\overline{x}_i\},\{x_i,v_{i+1}\},\{\overline{x}_i,v_{i+1}\}$.
    For each clause~$C_j$, let~$p_j,q_j,r_j$ be the three literals appearing in~$C_j$.
    We add the links~$\{c_j,p_j\},\{c_j,q_j\},\{c_j,r_j\},\{d_{j},p_j\},\{d_{j},q_j\},\{d_{j},r_j\}$.
    Finally, we add the links~$\{d_j,c_{j+1}\}$ and~$\{c_{j+1},d_j\}$ for each~$j < m$ and the links~$\{d_m,c_1\}$ and~$\{c_1,d_m\}$.
    An example of the constructed network is shown in \cref{fig:ideal}.
    \begin{figure}[t]
        \centering
        
        \begin{tikzpicture}
            \foreach \i in {1,2,3,4}{
                \node[circle,draw,label=$v_{\i}$] at(2*\i,0) (v\i) {};
                \node[circle,draw,label=$x_{\i}$] at(2*\i+1,1) (x\i) {} edge[-](v\i);
                \node[circle,draw,label=left:$\overline{x_{\i}}$] at(2*\i+1,-1) (n\i) {} edge[-](v\i) edge[-](x\i);
            }
            \node[circle,draw,label=$v_5$] at(10,0) (v5) {} edge[-](x4) edge[-](n4);
            \draw[-](x1) to (v2);
            \draw[-] (n1) to (v2);
            \draw[-] (x2) to (v3);
            \draw[-] (n2) to (v3);
            \draw[-] (x3) to (v4);
            \draw[-] (n3) to (v4);

            \node[circle,draw,label=below:$c_1$] at(8,-3) (c1) {} edge[-] (v5);
            \node[circle,draw,label=below:$c_2$] at(6,-3) (c2) {};
            \node[circle,draw,label=below:$c_3$] at(4,-3) (c3) {};
            \node[circle,draw,label=below:$d_1$] at(7,-3) (d1) {} edge[-](c2);
            \node[circle,draw,label=below:$d_2$] at(5,-3) (d2) {} edge[-](c3);
            \node[circle,draw,label=below:$d_3$] at(3,-3) (d3) {} edge[-,bend right=40](c1);

            \draw[-,bend right=15] (c1)  to (x1);
            \draw[-,bend left=10] (c1)  to (x2);
            \draw[-] (c1)  to (n4);
            \draw[-,bend left=20] (d1)  to (x1);
            \draw[-] (d1)  to (x2);
            \draw[-] (d1)  to (n4);

            \draw[-] (c2)  to (n1);
            \draw[-] (c2)  to (n3);
            \draw[-] (c2)  to (x4);
            \draw[-] (d2)  to (n1);
            \draw[-] (d2)  to (n3);
            \draw[-,bend right=10] (d2)  to (x4);

            \draw[-] (c3)  to (x1);
            \draw[-] (c3)  to (n2);
            \draw[-,bend right=10] (c3)  to (x3);
            \draw[-,bend left=10] (d3)  to (x1);
            \draw[-] (d3)  to (n2);
            \draw[-,bend right=10] (d3)  to (x3);
        \end{tikzpicture}
        \caption{Construction of the graph in the proof of \cref{thm:idealhard} for the formula \[\phi = (x_1 \lor x_2 \lor \neg x_4) \land (\neg x_1 \lor \neg x_3 \lor x_4) \land (x_1 \lor \neg x_2 \lor x_3)\] excluding the clique~$K$ and the target~$t$.}
        \label{fig:ideal}
    \end{figure}
    We continue with the priority lists and list them in \cref{tab:prio2}.
    Note that~$K$ ensures that the graph is~$(2n+2m+1)$-connected and~$t$ has exactly~$2n+2m+1$ incident links.
    Thus, we assume that at most~$k=2n+2m$ links fail and we only present the first~$k+1$ entries in each priority list.
    \begin{figure}[t]
        \centering
        \begin{tabular}{c|c|r}
             node & in-port & priority list \\\hline
             $u_i$ & $*$ & $t, u_{i+1}, \ldots u_{2n+2m},u_0,\ldots,u_{i-1}$\\\hline
             $v_1$ & $v_1$ & $x_1, \overline{x_1}, K$\\
             $v_1$ & $*$ & $K$\\
             $v_i$ & $x_{i-1}$ & $\overline{x_{i-1}}, x_i, \overline{x_i}, K$\\
             $v_i$ & $\overline{x_{i-1}}$ & $x_{i-1}, x_i, \overline{x_i}, K$\\
             $v_i$ & $*$ & $K$\\
              $v_{n+1}$ & $x_n$ & $\overline{x_n}, c_1, K$\\
             $v_{n+1}$ & $\overline{x_n}$ & $x_n, c_1, K$\\
             $v_{n+1}$ & $*$ & $K$\\
             \hline
             $x_i$ & $v_{i}$ & $\overline{x_i}, K$\\
             $x_i$ & $\overline{x_i}$ & $v_i, v_{i+1}, K$\\
             $x_i$ & $c_j$ & $v_i, d_j, K$\\
             $x_i$ & $*$ & $K$\\
             $\overline{x_i}$ & $v_i$ & $x_i, K$\\
             $\overline{x_i}$ & $x_i$ & $v_i, v_{i+1}, K$\\
             $\overline{x_i}$ & $c_j$ & $v_i, d_j, K$\\
             $\overline{x_i}$ & $*$ & $K$\\\hline
             $c_1$ & $v_{n+1},d_m$ & $p_1, q_1, r_1, K$\\
             $c_1$ & $*$ & $K$\\
             $c_j$ & $d_j$ & $p_j, q_j, r_j, K$\\
             $c_j$ & $*$ & $K$\\
             $d_m$ & $p_m,q_m,r_m$ & $c_1, K$\\
             $d_m$ & $*$ & $K$\\
             $d_j$ & $p_j,q_j,r_j$ & $c_{j+1},K$\\
             $d_j$ & $*$ & $K$
        \end{tabular}
        \caption{Table of the priority lists for all nodes (except for the target~$t$) in the construction in the proof of \cref{thm:idealhard}.
        The symbol $*$ stands for all in-ports that are not listed elsewhere in the table. We use~$p_i,q_i,$ and~$r_i$ to denote the three literals in clause~$C_i$, respectively. As an example, if clause~${C_1 = (x_1 \lor x_2 \lor \overline{x_4})}$, then~$p_1 = x_1$, $q_1 = x_2$, and~$r_1 = \overline{x_4}$. We only give the first~${k+1=2n+2m+1}$ entries in each priority list and~$K$ stands for the priority list~$u_0, u_1, \ldots, u_{2n+2m}$.}
        \label{tab:prio2}
    \end{figure}
    Before we formally state the proof, let us give some high-level intuition.
    We will later show that if the packet ever reaches a node in~$K$, then it reaches~$t$.
    Hence, for the packet to not reach~$t$, it has to start at node~$v_1$ as all other nodes immediately send it to some node in~$K$ (assuming at most~$k$ links fail).
    Next, we will show that for each~$i \in [n]$, if the packet reaches~$v_i$, then the only way to not be sent to a node in~$K$ is if either the links~$\{v_i,x_i\}$ and~$\{\overline{x_i},v_{i+1}\}$ or the links~$\{v_i,\overline{x_i}\}$ and~$\{x_i,v_{i+1}\}$ fail.
    In both cases, the packet reaches~$v_{i+1}$ and which of the links incident to~$v_i$ fails will encode the assignment of variable~$x_i$ in a solution.
    More specifically, we say in the first case that~$x_i$ is set to true and in the second case that it is set to false.
    Once the packet reaches~$v_{n+1}$, it is sent to~$c_1$ (or to~$K$ and subsequently to~$t$).
    Now, for each~$c_j$, if the three incident links to~$p_j,q_j$ and~$r_j$ fail, then the packet is sent to~$K$.
    So at least one of these links does not fail and when the packet is sent to one of these three nodes, then it is checked that the respective variable is set correctly to satisfy~$C_j$.
    If this is not the case, then the packet is sent to some~$v_i$, then to some node in~$K$ and subsequently to~$t$.
    So the only way to avoid this is to send the packet to some variable node satisfying the respective clause.
    In this case, the packet is forwarded to~$d_j$ and then to~$c_{j+1}$.
    Once the node~$d_m$ is reached (which corresponds to the case where the entire formula is satisfied by the chosen assignment), the routing continues with~$c_1$ and then repeats the cycle indefinitely.

    We now make the argument more formal.
    First, we show that if some node in~$K$ is reached and at most~$k$ links fail, then the target~$t$ is reached.
    Let~$i$ be the index such that~$u_i \in K$ is reached and assume towards a contradiction that there is a failure scenario~$F$ of size at most~$k=2m+2n$ such that~$t$ is not reached.
    By construction, if some node~$u_j \in K$ is reached and the link~$\{u_j,t\}$ does not fail, then~$t$ is reached.
    Hence,~$\{u_i,t\} \in F$.
    We will consider two indices~$a$ and~$b$, where initially~$a=i$ and~$b=i+1$.
    The node~$u_a$ is the last node reached in~$K$ and~$u_b$ is the next node to be visited, that is, the links~$\{u_a,u_j\}$ fail for all~$j \in \{a+1,a+2,\ldots,b-1\}$ if~$b > a$ and for all~$j \in \{a+1,a+2,\ldots,2n+2m,0,1,\ldots,b-1\}$ otherwise.
    If none of the links~$\{u_i,u_{i+1}\}$ and~$\{u_{i+1},t\}$ fail, then~$t$ is reached.
    If~$\{u_i,u_{i+1}\}$ does not fail, then we update~$u_a = u_b$.
    Otherwise, we keep~$u_a$ unchanged.
    In either case, we update~$u_b = u_{b+1}$ if~$b \neq 2n+2m$ and~$u_b = u_0$ otherwise.
    Note that for each node~$u_b$ with~$u \neq i$ it holds that some link~$\{u_a,u_b\}$ with~$a \in \{i,i+1,\ldots,b-1\}$ if~$i < b$ or~$a \in \{i,i+1,\ldots,2n+2m,0,1,\ldots,b-1\}$ or the link~$\{u_b,t\}$ fails.
    Since the same link is never counted for two different indices~$b$ (if~$a$ is between~$i$ and~$b$, then~$b$ cannot be between~$i$ and~$a$), this gives~$2n+2m$ links (in addition to~$\{u_i,b\}$) in~$F$, a contradiction to~${|F| \leq 2n+2m = k}$.
    Thus, if a node in~$K$ is reached and at most~$k$ links fail, then~$t$ is reached.

    To conclude the proof, note that if~$\phi$ is satisfiable, then let~$\beta$ be a satisfying assignment.
    For each variable~$x_i$, if~$\beta$ sets it to true, then we let links~$\{v_i,x_i\}$ and~$\{\overline{x_i},v_{i+1}\}$ fail and otherwise, we let links~$\{v_i,\overline{x_i}\}$ and~$\{x_i,v_{i+1}\}$ fail.
    Next, for each clause~$C_j$, we let the links between~$c_j$ and all literals that do not satisfy~$C_j$ under~$\beta$ fail.
    Since each clause has size three and~$\beta$ satisfies~$\phi$, this fails at most~$2m$ clauses.
    Now, as shown above, if the routing starts in~$v_1$, then it is routed through all variables and reaches~$v_{n+1}$.
    From there, it is forwarded to~$c_1$.
    Next, for each clause, it is forwarded to one literal such that~$C_j$ is satisfied by the respective literal under~$\beta$.
    Note that this corresponds to the case where the link between that literal ($x_i$ or~$\overline{x_i}$) and~$v_i$ fails.
    Hence, the packet is next sent to~$d_j$ and then to~$c_{j+1}$ (or from~$d_m$ to~$c_1$).
    Since this cycle continues indefinitely, the constructed forwarding pattern is not ideally resilient.

    In the other direction, if the constructed forwarding pattern is not ideally resilient, then the routing has to start in~$v_1$ and avoid~$K$ as shown above.
    Hence, for each variable~$x_i$ it holds that link~$\{v_i,x_i\}$ and~$\{\overline{x_i},v_{i+1}\}$ or~$\{v_i,\overline{x_i}\}$ and~$\{x_i,v_{i+1}\}$ fail.
    We define an assignment for all variables by setting~$x_i$ to true whenever~$\{v_i,x_i\}$ fails and to false whenever~$\{v_i,\overline{x_i}\}$ fails.
    Note that not both of the links can fail as otherwise~$v_i$ sends the packet to some node in~$K$.
    Now, the only possibility for the routing to continue indefinitely is if the packet is sent through a cycle of all nodes~$c_j$ and~$d_j$.
    Between these two nodes, it has to be sent through a node~$x_i$ or~$\overline{x_i}$.
    Once this nodes receives the packet from~$c_j$, it sends it to~$v_i$ if the link~$\{x_i,v_i\}$ does not fail.
    The node~$v_i$ then sends it to some node in~$K$.
    Hence, the link~$\{x_i,v_i\}$ fails and this corresponds to the case where~$C_j$ is satisfied by~$x_i$ in the constructed assignment.
    This assignment satisfies all clauses, the input formula~$\phi$ is satisfiable and that concludes the proof.
\end{proof}

We mention that nodes of high degree are necessary in the case of ideal resilience as any constant-degree node means that the connectivity of the input graph and therefore also~$k$ is constant which results in a polynomial-time solvable problem.
We leave it as an open problem whether the reduction for perfect resilience can be modified to have bounded degree.

Note that the number of nodes in the networks constructed in the previous proofs are linear in the number of variables and clauses of the input formulas.
Assuming a complexity hypothesis called the exponential time hypothesis (ETH), \textsc{3-Sat} cannot be solved in~$2^{o(n+m)}$ time~$\cite{IPZ01}$, where~$n$ and~$m$ are the number of variables and clauses in the input formula.
Hence, both verification problems cannot be solved in~$2^{o(n)}$ time.
This contrasts the typical results
that many problems on planar graphs can be solved in subexponential time (usually~$2^{\sqrt{n}}$~time)~\cite{DH08,KM14,Ned20}.
However, as we just proved, this is surprisingly not the case for the verification problem for ideal or perfect resilience (assuming the ETH).

\section{Conclusion and Open Questions}\label{sec:conclusion}

Fast re-route mechanisms that allow for a local and immediate reaction to link failures are essential in present-day dependable computer networks. Several approaches for
fast re-routing have been implemented and deployed in essentially every type of computer network. However, the
question for which graphs we can efficiently construct perfectly resilient protection mechanisms, meaning that
they should guarantee packet delivery in any failure scenario as long as the source and target nodes remain 
physically connected, has not been solved yet. 
The exact computational complexity 
of deciding
whether a given network topology together with a target node allows for a perfectly resilient routing protection
is showed to be polynomial by the application of Robertson and Seymour theorem~\cite{RS-theorem} together
with the minor-stability result~\cite{FHPST22} for the existence of a perfectly resilient routing to
all possible target nodes as well as to
a given target node (via the rooted variant of this
result~\cite{RS90,RS95}).
This contrasts our result that the verification problem
is computationally hard---we showed by a nontrivial reduction from \textsc{3-Sat} that the question whether a given set of routing tables for a fixed target node
is perfectly resilient is \coNP-complete.

We have also studied a variant of a local fast re-route that is in-port oblivious (i.e., the packet's incoming
interface is not part of the routing tables). This potentially enables a faster failover protection as well as a
more memory-efficient way of storing the forwarding tables. In this restricted case, we were able to
provide a linear-time algorithm that determines whether a given in-port oblivious routing is 
perfectly resilient and we showed that the synthesis problem is also decidable in linear time, allowing
us to efficiently construct in-port oblivious routing tables 
for the topologies where this is indeed possible.

A major open problem is to design a concrete algorithm that 
for a given network constructs in polynomial time perfectly resilient routing tables whenever this is possible.

\section*{Acknowledgements}
This work was partially supported by the European Research Council~(ERC) under the European Union’s Horizon 2020 research and innovation program (grant agreement No. 819416) and ERC Consolidator grant AdjustNet (agreement No. 864228).

\bibliographystyle{plainnat}
\bibliography{references}

\end{document}